\documentclass[12pt,a4paper]{article}
	\usepackage{amssymb}
    \usepackage{amsmath}
    \usepackage{amsthm}
   	\usepackage{multirow}
    \usepackage{color}
	\usepackage{tikz}
	\usepackage{authblk}
	\usepackage{setspace}
	\usepackage{bibspacing}

	\providecommand{\keywords}[1]{\textbf{\textit{Keywords: }} #1}

	\title{\bf Minimum guesswork discrimination between quantum states}
	\author[1,2,3]{\normalsize Weien Chen\footnote{chenweienn@gmail.com}}
	\author[1,2]{Yongzhi Cao}
	\author[1,2]{Hanpin Wang}
	\author[3]{Yuan Feng\footnote{Yuan.Feng@uts.edu.au}}
	
		\affil[1]{\footnotesize Institute of Software, School of EECS, Peking University, China}
		\affil[2]{\footnotesize Key Laboratory of HCST, Ministry of Education, China}
		\affil[3]{\footnotesize Centre for QCIS, University of Technology, Sydney, Australia}

	\theoremstyle{plain}
        \newtheorem{mythm}{Theorem}
        \newtheorem{myprop}{Proposition}
        \newtheorem{mycoro}{Corollary}
        \newtheorem{mylemma}{Lemma}
	\theoremstyle{definition}
        
        \newtheorem{myexa}{Example}
    \theoremstyle{remark}
        \newtheorem{myrmk}{Remark}

\begin{document}

	\maketitle

	\begin{abstract}
		Error probability is a popular and well-studied optimization criterion in discriminating non-orthogonal quantum states. It captures the threat from an adversary who can only query the actual state once. However, when the adversary is able to use a brute-force strategy to query the state, discrimination measurement with minimum error probability does not necessarily minimize the number of queries to get the actual state. In light of this, we take Massey's guesswork as the underlying optimization criterion and study the problem of minimum guesswork discrimination. We show that this problem can be reduced to a semidefinite programming problem. Necessary and sufficient conditions when a measurement achieves minimum guesswork are presented. We also reveal the relation between minimum guesswork and minimum error probability. We show that the two criteria generally disagree with each other, except for the special case with two states. Both upper and lower information-theoretic bounds on minimum guesswork are given. For geometrically uniform quantum states, we provide sufficient conditions when a measurement achieves minimum guesswork. Moreover, we give the necessary and sufficient condition under which making no measurement at all would be the optimal strategy.
	\end{abstract}
	\keywords{quantum state discrimination, error probability, guesswork, brute-force strategy, information flow}		
	
	\section{Introduction}
	\noindent
		\label{sec_introduction}
		Since Helstrom's pioneering work on quantum binary decision problem~\cite{Helstrom1976}, quantum state discrimination has been extensively studied~\cite{Chefles2000,Barnett2009,Helstrom1976,Davies1978,Schumacher1995,Fuchs1996,BanYH1997,OsakiHB1998,Sasaki1999,Ivanovic1987,Dieks1988,Peres1988,EldarF2001,Holevo1973,Yuen1975,BanKMH1997,Barnett2001,EldarMV2004,BaeH2013}. The problem is usually described as a protocol between two parties, conventionally named Alice and Bob. Alice selects a quantum state $\rho$ from a set $\{\rho_{i}\}$ according to a probability distribution $\{p_i\}$ and gives it to Bob. We assume that Bob knows both the set of possible states and their associated probabilities. His aim is to identify the actual prepared state. To this end, Bob performs some quantum measurement on $\rho$ in order to extract information about the index $i$. This gives rise to an optimization problem with regard to Bob's choice of measurement. A number of criteria have been considered to concretize the meaning of this optimality~\cite{Helstrom1976,Davies1978,Schumacher1995,Fuchs1996,BanYH1997,OsakiHB1998,Sasaki1999,Ivanovic1987,Dieks1988,Peres1988,EldarF2001,Holevo1973,Yuen1975,BanKMH1997,Barnett2001,EldarMV2004,CrokeABGJ2006}, among which error probability and the Shannon mutual information are two representatives. While the former has led to a research line known as \emph{minimum error discrimination} (MED)~\cite{Holevo1973,Yuen1975,BanKMH1997,Barnett2001,EldarMV2004,BaeH2013}, the later corresponds to the study of \emph{accessible information}~\cite{Holevo1973b,JozsaRW1994,BanYH1997,OsakiHB1998,Sasaki1999}. Interestingly, as an alternative to MED, an unambiguous (error-less) scheme of state discrimination has been proposed, by allowing certain fraction of inconclusive measurement outcomes~\cite{Ivanovic1987,Dieks1988,Peres1988}.  

		We point out that quantum state discrimination can be seen as a special case of \emph{quantitative information flow} (QIF) analysis, which has been an active topic in security community during the last decades~\cite{Millen1987,Mclean1990,Gray1992,Lowe2002,KopfB2007,ChatziPP2008,Smith2009,Smith2011,Alvim2012,Alvim2014}. In QIF analysis, the aim is to quantify the amount of information leaked by a covert channel from a high-level entity, whose secret information (e.g., a password) is mathematically described as a random variable $X$ with alphabet $\{x_i\}$ and the associated probability distribution $\{p(x_i)\}$, to a low-level entity, whose partial information about $X$ is described as another random variable $Y$ with alphabet $\{y_j\}$. The correlation between $X$ and $Y$ is determined by the channel matrix $\{p(y_j|x_i)\}$ of the covert channel. To put quantum state discrimination in the context of QIF analysis, we may view Alice as the high-level entity and Bob the low-level entity. The only restriction is that the correlation between these two entities are ruled by quantum mechanics: Alice encodes her classical secret messages $\{x_i\}$ into quantum states $\{\rho_{x_i}\}$; Bob performs a measurement on Alice's prepared state $\rho_{x}$ to get information about $X$ and stores his measurement outcome in $Y$; the channel matrix is then given by the Born rule~\cite{Gleason1957}, $p(y_j|x_i) = \mathrm{Tr}(\rho_{x_i}\pi_{y_j})$, where $\pi_{y_j}$ is the measurement operator corresponding to the outcome $y_j$.
		
		In the literature of QIF analysis, researchers have proposed different figures of merit to quantify how successfully a low-level entity Bob can identify the secret value of $X$ given knowledge about $Y$, according to different adversarial strategies which Bob may adopt. In particular, it is well-known that error probability, guesswork, and the Shannon entropy deal with one-shot strategy, brute-force strategy, and subset membership strategy, respectively, and thus play important and complementary roles in QIF analysis~\cite{Cachin1997,Cover2006,Alvim2010}. In the quantum setting, it is clear that one-shot strategy and subset membership strategy have been considered. Error probability and the Shannon entropy have been widely studied in quantum information theory, and led a large amount of research on, besides MED and accessible information discussed above, quantum source coding~\cite{Schumacher1995,JozsaS1994}, quantum channel capacity~\cite{HausladenJSWW1996,Holevo1998,SchW1997}, etc. However, to the best of our knowledge, no work has addressed brute-force strategy in the context of quantum state discrimination.

		The above observation motivates us to consider Massey's guesswork~\cite{Massey1994} as the optimization criterion in quantum state discrimination. We name the new problem \emph{minimum guesswork discrimination} (MGD). In contrast to the MED scenario where Bob has only one chance to ask Alice ``is $\rho_x = \rho_{x_i}$'' for some $x_i$ chosen based on his measurement outcome, in this study Bob carries out multiple such queries until hitting Alice's prepared state $\rho_x$. Guesswork, the new criterion, quantifies the expected number of queries that Bob needs to make. We hope our preliminary step towards the study of brute-force strategy in the quantum setting will initiate a sibling direction of MED.

		The rest of this paper is organized as follows. In Section~\ref{sec_qgp}, we first review the classical guessing problem, then extend it to the quantum setting. We show that MGD can be reduced to a semidefinite programming (SDP) problem, and present necessary and sufficient conditions  which must be satisfied by the optimal measurement to achieve minimum guesswork. Section~\ref{sec_g_e} is devoted to the relation between the minimum error criterion and the minimum guesswork criterion. We provide both upper and lower information-theoretic bounds on minimum guesswork in Section~\ref{sec_bound}, and sufficient conditions when a measurement achieves minimum guesswork for geometrically uniform states in Section~\ref{sec_gu}. In Section~\ref{sec_n_m}, we answer the question ``when would making no measurement at all be the optimal strategy?'' by a sufficient and necessary condition on the quantum ensemble. We discuss several other interesting issues worthy of consideration in Section~\ref{sec_discussion} and conclude this work in Section~\ref{sec_conclusion}.

	\section{Quantum guessing problem}
	\noindent
		\label{sec_qgp}
		We first review the classical guessing problem. Then its quantum variant can be simply formalized by instantiating the classical problem with quantum information and mechanics. Suppose that Alice has a discrete random variable $X$ with alphabet $\mathcal{X} = \{x_i: 1\leq i\leq n\}$ and the associated probability distribution $\{\mathrm{Pr}(X=x_i) \triangleq p(x_i): 1\leq i\leq n\}$. Bob, who knows both the alphabet and the distribution, aims to identify the true value of $X$ by keeping on asking questions of the form ``is $X = x_i$ ?'' until getting the answer ``yes''. How many guesses is he expected to make? Massey~\cite{Massey1994} observed that Bob's optimal strategy for minimizing his work is to arrange his queries according to the non-increasing order of probabilities $p(x_i)$'s. Formally, the \emph{guesswork} of Bob is given by
			\begin{align}
				\label{eq_def_guesswork}
				\mathrm{G(X)} \triangleq \sum_{i=1}^n \sigma(i) p(x_i),
			\end{align}
		where $\sigma$ is a permutation on the index set $\{1,\cdots,n\}$ such that $p(x_i)\geq p(x_j)$ implies $\sigma(i)\leq \sigma(j)$. Recall that a permutation on set $S$ is just an one-to-one mapping from $S$ to itself. Here, $\sigma$ represents formally Bob's guessing strategy: he guesses $x_i$ in his $\sigma(i)$th query. The guesswork $\mathrm{G}(X)$ quantifies the expected number of queries that Bob needs to guess the actual value of $X$ when applying the strategy $\sigma$. The constraint on $\sigma$ ensures its optimality in that any other permutation yields greater or equal guesswork.

		The above definition of guesswork has been generalized to conditional version~\cite{Arikan1996}, which is more appealing in practice. In addition to Alice's alphabet $\mathcal{X}$ and the prior distribution on it, Bob may possess some extra knowledge (or side information) of $X$. In general, we assume a channel existing between Alice and Bob with input set $\mathcal{X}$ and output set $\mathcal{Y} = \{y_j: 1\leq j\leq m\}$. The probabilistic behavior of this channel is characterized in the standard way, i.e., by conditional probabilities $\{p(y_j|x_i): 1\leq i\leq n, 1\leq j\leq m \}$. Consequently, given some fixed input random variable $X$ of Alice, we can derive an output random variable $Y$ on Bob's side with the associated distribution obtained by
			\begin{align*}
				p(y_j) = \sum_{i=1}^n p(x_i) p(y_j|x_i),
			\end{align*}
		for each $1\leq j\leq m$. As usual, we denote the joint distribution of $X$ and $Y$ by $\{p(x_i,y_j)\}$.

		Now, instead of consulting the prior distribution of $X$ which leads to the guesswork $\mathrm{G}(X)$, Bob applies an optimal guessing strategy on each posterior distribution $\{p(x_i|y_j): 1\leq i\leq n\}$ when $y_j$ is observed. We denote each corresponding posterior guesswork by $\mathrm{G}(X|Y = y_j)$. Bob's conditional guesswork is then given by
			\begin{equation}
				\label{eq_def_cond_guesswork}
				\begin{split}
				\mathrm{G}(X|Y) &\triangleq \sum_{j=1}^m p(y_j) \mathrm{G}(X|Y = y_j) \\
								&= \sum_{j=1}^m p(y_j) \sum_{i=1}^n \sigma_j(i) p(x_i|y_j),
				\end{split}
			\end{equation}
		where each $\sigma_j$ is a permutation on $\{1,\cdots,n\}$ such that $p(x_i|y_j)\geq p(x_{i'}|y_j)$ implies $\sigma_j(i)\leq \sigma_j(i')$. In~\cite{Arikan1996}, Arikan showed that extra knowledge always reduces (at least preserves) guesswork, i.e., the inequality $\mathrm{G}(X|Y) \leq \mathrm{G}(X)$ holds.


		We now define the quantum guessing problem by extending the above scenario to the quantum setting. 
		Alice selects from her alphabet a secret message $x_i$ with probability $p(x_i)$ and encodes it into a (possibly mixed) quantum state $\rho_{x_i}$, which is accessible to Bob. Alice's operation gives rise to an ensemble of quantum states $\mathcal{E} = \{(p(x_i), \rho_{x_i})\}$ living in a finite dimensional Hilbert space $\mathcal{H}$. We call this $\mathcal{E}$ a \emph{quantum encoding} of $X$. In order to identify Alice's secret message, Bob performs on the quantum state a \emph{positive operator-valued measure} (POVM) $\Pi = \{\pi_{y_j}: 1\leq j\leq m\}$, which comprises $m$ positive semidefinite (PSD) operators satisfying the completeness condition $\sum_{j=1}^m \pi_{y_j} = I_{\mathcal{H}}$, where $I_{\mathcal{H}}$ is the identity matrix in $\mathcal{H}$.  The probability $p(y_j|x_i)$ that Bob obtains the $j$th measurement outcome when Alice chooses the $i$th message is given by $\mathrm{Tr}(\rho_{x_i}\pi_{y_j})$. Note that the random variable $Y$ is completely determined by the ensemble $\mathcal{E}$ and the POVM $\Pi$. Hence, minimizing over all possible POVMs leads to the following definition of \emph{minimum guesswork}: 
			\begin{align}
				\label{eq_def_min_guesswork}
				\mathrm{G}^{opt}(\mathcal{E}) \triangleq \min_{\Pi\in\mathcal{M}} \mathrm{G}(X|Y),
			\end{align}
		where $\mathcal{M}$ is the set of all POVMs. We name this minimization problem \emph{minimum guesswork discrimination} (MGD).
		
		For convenience of the following reasoning, we introduce some notations. Let $\mathcal{P}$ be the set of all non-zero PSD operators and $\mathcal{P}_1$ be the set of all non-zero rank-one PSD operators. A \emph{complete} POVM is a POVM comprising only rank-one PSD operators. The set of all complete POVMs is denoted by $\mathcal{M}_c$. Given $\pi\in\mathcal{P}$, the random variable $X_{\pi}$ which takes value from the alphabet of $X$ is defined by
			\begin{align*}
				\mathrm{Pr}(X_{\pi} = x_i) \triangleq \frac{ p(x_i)\mathrm{Tr}(\rho_{x_i}\pi) }{ \sum_{i=1}^n p(x_i)\mathrm{Tr}(\rho_{x_i}\pi) }.
			\end{align*}
		Intuitively, $X_{\pi}$ describes Bob's posterior distribution over Alice's messages when he obtains the outcome indicated by the measurement operator $\pi$. With this notation, the guesswork $\mathrm{G}(X|Y)$ can be rewritten as $\sum_{j=1}^m p(y_j)\mathrm{G}(X_{\pi_{y_j}})$ for some $\mathcal{E}$ and $\Pi$ as given in the preceding paragraph. Sometimes, we write $\mathrm{G}(X|\Pi)$ instead of $\mathrm{G}(X|Y)$ to indicate the specific POVM adopted by Bob.
		
		In the following, we give several alternative characterizations of $\mathrm{G}^{opt}(\mathcal{E})$. The first one states that the optimal POVM achieving $\mathrm{G}^{opt}(\mathcal{E})$ can always be taken as a complete measurement.
			\vspace*{12pt}
			\begin{myprop}
				\label{thm_cm}
				Let $\mathcal{E}$ be a quantum encoding of a random variable $X$, and $\mathrm{G}^{opt}(\mathcal{E})$ be defined in Eq.(\ref{eq_def_min_guesswork}). It holds that
					\begin{align*}
						\mathrm{G}^{opt}(\mathcal{E}) = \min_{\Pi\in\mathcal{M}_c} \mathrm{G}(X|Y).
					\end{align*}
			\end{myprop}
			\begin{proof}
				Since $\mathrm{G}(X|Y)$ can be rewritten as $\sum_{j=1}^m p(y_j)\mathrm{G}(X_{\pi_{y_j}})$ for a POVM $\Pi = \{\pi_{y_j}\}$, it is sufficient to prove that for any $\pi_{y_j}\in\Pi$ the guesswork $\mathrm{G}(X|Y)$ cannot be increased by splitting $\pi_{y_j}$ into two measurement operators $\pi_{y_{j'}}$ and $\pi_{y_{j''}}$ such that $\pi_{y_j} = \pi_{y_{j'}} + \pi_{y_{j''}}$. Formally, we have the following inference:
					\begin{equation}
						\label{eq_proof_thm_cm}
						\begin{split}
						p(y_j)\mathrm{G}(X_{\pi_{y_{j}}}) & = p(y_j) \sum_{i=1}^n \sigma_j(i) \frac{p(x_i) \mathrm{Tr}(\rho_{x_i}\pi_{y_j})}{p(y_j)} \\
							& = \sum_{i=1}^n \sigma_j(i) p(x_i) \mathrm{Tr}(\rho_{x_i}\pi_{y_{j'}}) \\
							& \quad + \sum_{i=1}^n \sigma_j(i) p(x_i) \mathrm{Tr}(\rho_{x_i}\pi_{y_{j''}}) \\
							& \geq \sum_{i=1}^n \sigma_{j'}(i) p(x_i) \mathrm{Tr}(\rho_{x_i}\pi_{y_{j'}}) \\
							& \quad + \sum_{i=1}^n \sigma_{j''}(i) p(x_i) \mathrm{Tr}(\rho_{x_i}\pi_{y_{j''}}) \\
							& = p(y_{j'}) \mathrm{G}(X_{\pi_{y_{j'}}}) + p(y_{j''}) \mathrm{G}(X_{\pi_{y_{j''}}}),
						\end{split}
					\end{equation}
				where $p(y_j) = \sum_{i=1}^n p(x_i)\mathrm{Tr}(\rho_{x_i}\pi_{y_j})$ and the inequality is due to the fact that the permutation $\sigma_j$ may not necessarily be optimal with respect to $X_{\pi_{y_{j'}}}$ or $X_{\pi_{y_{j''}}}$. Therefore, from an optimal POVM $\Pi$ which achieves $\mathrm{G}^{opt}(\mathcal{E})$, one can always derive a complete POVM which also achieves $\mathrm{G}^{opt}(\mathcal{E})$.
			\end{proof}
		
		Another alternative characterization of $\mathrm{G}^{opt}(\mathcal{E})$ is suggested by the minimum error criterion used in MED. As shown in Remark~\ref{rmk_p} in Section~\ref{sec_g_e}, to minimize error probability it suffices to consider POVMs $\Pi$'s which satisfy $|\Pi| = n$ when $\mathcal{E}$ comprises $n$ states. Here, $|\Pi|$ denotes the number of measurement operators in $\Pi$. Similarly, in MGD, the following simple observation shows an $n!$ upper bound on $|\Pi|$ in minimizing guesswork. Suppose $\pi_1$ and $\pi_2$ are two measurement operators in $\Pi$. If the two posterior distributions induced by $\pi_1$ and $\pi_2$ achieve their minimum guesswork $\mathrm{G}(X_{\pi_1})$ and $\mathrm{G}(X_{\pi_2})$ by the same optimal permutation $\sigma$, then merging $\pi_1$ and $\pi_2$ into a single measurement operator $\pi_1+\pi_2$ would not change the overall performance of $\Pi$. Since there are exactly $n!$ different permutations on the index set $\{1,2,\cdots,n\}$, we can reduce $|\Pi|$ to $n!$ for an optimal POVM $\Pi$.
			\vspace*{12pt}
			\begin{myprop}
				\label{prop_n!}
				Let $\mathcal{E}$ be a quantum encoding of a random variable $X$, and $\mathrm{G}^{opt}(\mathcal{E})$ be defined in Eq.(\ref{eq_def_min_guesswork}). It holds that
					\begin{align*}
						\mathrm{G}^{opt}(\mathcal{E}) = \min_{\Pi\in\mathcal{M}_{n!}} \mathrm{G}(X|Y),
					\end{align*}
				where $\mathcal{M}_{n!}$ is the set of all POVMs consisting of exactly $n!$ measurement operators.
			\end{myprop}
			\vspace*{12pt}
		
		Based on Eldar et al.'s analogous results in MED~\cite{EldarMV2003}, we reduce MGD to an SDP problem, which has numerical solutions within any desired accuracy in mathematics, and derive necessary and sufficient conditions satisfied by the optimal POVM to achieve minimum guesswork. We present the results below and omit the proof which is simply an imitation of the reasoning in~\cite{EldarMV2003}.
			\vspace*{12pt}
			\begin{myprop}
				\label{prop_sdp}
				Let $\mathcal{E}$ be a quantum encoding of a random variable $X$, and $\mathrm{G}^{opt}(\mathcal{E})$ be defined in Eq.(\ref{eq_def_min_guesswork}). It holds that
					\begin{align*}
						\mathrm{G}^{opt}(\mathcal{E}) = \max_{A} \mathrm{Tr}(A),
					\end{align*}
				where $A$ ranges over all Hermitian operators in $\mathcal{H}$ satisfying $A\leq \sum_{i=1}^n\sigma(i)p(x_i)\rho_{x_i}$ for any permutation $\sigma$ on $\{1,\cdots,n\}$.
			\end{myprop}
			\begin{myprop}
			\vspace*{12pt}
				\label{prop_ns}
				Let $\mathcal{E}$ be a quantum encoding of a random variable $X$, and $\mathrm{G}^{opt}(\mathcal{E})$ be defined in Eq.(\ref{eq_def_min_guesswork}). A POVM $\{\pi_{y_1},\pi_{y_2},\cdots,\pi_{y_{n!}}\}$ achieves $\mathrm{G}^{opt}(\mathcal{E})$ if and only if for any permutation $\sigma$ on $\{1,\cdots,n\}$ it holds that
					\begin{align*}
						\sum_{i=1}^n\sum_{j=1}^{n!} \sigma_j(i)p(x_i)\rho_{x_i}\pi_{y_j} \leq \sum_{i=1}^n \sigma(i)p(x_i)\rho_{x_i}.
					\end{align*}
			\end{myprop}
			\vspace*{12pt}
		It is worth noting that Proposition~\ref{prop_ns} can also be proved directly using the technique introduced in~\cite{BarnettC2009}.

	\section{Relation between MGD and MED}
	\noindent
		\label{sec_g_e}
		Given two random variables $X$ and $Y$, the unconditional and conditional error probabilities of guessing the value of $X$ are defined respectively by
			\begin{equation}
				\label{eq_def_p}
				\mathrm{P}_{err}(X)  \triangleq 1 - \max_{1\leq i\leq n} p(x_i),
			\end{equation}
		and
			\begin{equation}
				\label{eq_def_cond_p}
				\begin{split}
				\mathrm{P}_{err}(X|Y) & \triangleq \sum_{j=1}^m p(y_j) \mathrm{P}_{err}(X|Y=y_j) \\
									  & = 1 - \sum_{j=1}^m p(y_j) \max_{1\leq i\leq n} p(x_i|y_j).
				\end{split}
			\end{equation}
		The underlying observation here is that to minimize the error probability, Bob should always guess the most probable message according to his prior or posterior distributions. In addition, if $X$ and $Y$ are correlated by some quantum encoding $\mathcal{E}$ and some POVM $\Pi$ as in the preceding section, Bob's error probability can be written as
			\begin{align*}
				\mathrm{P}_{err}(X|Y) = 1 - \sum_{j=1}^m \max_{1\leq i\leq n} p(x_i) \mathrm{Tr}(\rho_{x_i}\pi_{y_j}),
			\end{align*}
		because of the fact that $p(y_j|x_i) = \mathrm{Tr}(\rho_{x_i}\pi_{y_j})$. The definition of minimum error probability in MED is then given by
			\begin{align}
				\label{eq_def_min_p}
				\mathrm{P}^{opt}_{err}(\mathcal{E}) \triangleq \min_{\Pi\in\mathcal{M}} \mathrm{P}_{err}(X|Y).
			\end{align}

			\vspace*{12pt}
		\begin{myrmk}
			\label{rmk_p}
			It is worth noting that in the literature of MED, error probability is usually formulated by 
				\begin{align*}
					\mathrm{P}'_{err}(X|Y) \triangleq 1- \sum_{i=1}^n p(x_i) \mathrm{Tr}(\rho_{x_i}\pi_{y_i}),
				\end{align*}
			with the POVM $\Pi$ being restricted to comprise exactly $n$ measurement operators.
			It is easy to show that this simplified characterization is equivalent to $\mathrm{P}_{err}(X|Y)$ in the sense that
				\begin{align*}
					\min_{\Pi\in\mathcal{M}_n} \mathrm{P}'_{err}(X|Y) = \min_{\Pi\in\mathcal{M}} \mathrm{P}_{err}(X|Y),
				\end{align*}
			where $\mathcal{M}_n$ is the set of all POVMs consisting of exactly $n$ measurement operators.
		\end{myrmk}
			\vspace*{12pt}

		Now, we are ready to investigate the relation between $\mathrm{G}^{opt}(\mathcal{E})$ and $\mathrm{P}^{opt}_{err}(\mathcal{E})$. To this end, we start by examining the two notions $\mathrm{G}(X)$ and $\mathrm{P}_{err}(X)$ in classic setting.

			\vspace*{12pt}
			\begin{mylemma}[\cite{Santis2001}, Lemma~2.4]
				\label{lem_gp1}
				Let $\mathrm{G}(X)$ and $\mathrm{P}_{err}(X)$ be defined in Eq.(\ref{eq_def_guesswork}) and Eq.(\ref{eq_def_p}), respectively. It holds that
					\begin{align}
						\label{eq_lem_gp1}
						\mathrm{G}(X) \leq \frac{n}{2}\cdot \mathrm{P}_{err}(X) + 1.
					\end{align}
			\end{mylemma}
			\vspace*{12pt}
		If we assume without loss of generality that $p(x_1) = \max_{i}p(x_i)$, then Eq.(\ref{eq_lem_gp1}) achieves equality when $p(x_i) = (1-p(x_1))/(n-1)$ for $2\leq i\leq n$. It is straightforward to prove the conditional counterpart of Eq.(\ref{eq_lem_gp1}).
			\vspace*{12pt}
			\begin{mycoro}
				\label{coro_gp1}
				Let $\mathrm{G}(X|Y)$ and $\mathrm{P}_{err}(X|Y)$ be defined in Eq.(\ref{eq_def_cond_guesswork}) and Eq.(\ref{eq_def_cond_p}), respectively. It holds that
					\begin{align*}
						\mathrm{G}(X|Y) \leq \frac{n}{2}\cdot \mathrm{P}_{err}(X|Y) + 1.
					\end{align*}
			\end{mycoro}
			\vspace*{12pt}
		Interestingly, guesswork can also be lower bounded in terms of error probability.
			\vspace*{12pt}
			\begin{mylemma}
				\label{lem_gp2}
				Let $\mathrm{G}(X)$ and $\mathrm{P}_{err}(X)$ be defined in Eq.(\ref{eq_def_guesswork}) and Eq.(\ref{eq_def_p}), respectively. It holds that
					\begin{align}
						\label{eq_lem_gp2}
						\mathrm{G}(X) \geq \frac{1}{2( 1 - \mathrm{P}_{err}(X) )} + \frac{1}{2},
					\end{align}
				with the equality achieved when $k$ probabilities in $\{p(x_i)\}$ are equal to $1/k$ for some integer $k\ (1\leq k\leq n)$ and the other probabilities all equal to zero.
			\end{mylemma}
			\begin{proof}
				Without loss of generality, we assume that $p(x_1) \geq \cdots \geq p(x_n)$. Given $p(x_1)$ fixed, in order to minimize $\mathrm{G}(X)$ we need to require as many probabilities $p(x_i)$'s being equal to $p(x_1)$ as possible. It follows that
					\begin{align*}
						\mathrm{G}(X) & \geq \sum_{i=1}^{k} i\cdot p(x_1) + (k + 1)\cdot (1 - p(x_1)\cdot k) \\
									& = -\frac{p(x_1)}{2} k^2 - \frac{p(x_1)}{2} k + k + 1 \\
									& \geq \frac{1}{2 p(x_1)} + \frac{1}{2} \\
									& = \frac{1}{2(1-\mathrm{P}_{err}(X))} + \frac{1}{2},
					\end{align*}
				where $k = \left\lfloor \frac{1}{p(x_1)} \right\rfloor$. The only non-trivial part of the above reasoning is the second inequality. To prove it, let $t = 1/p(x_1) - k$. Then, it is equivalent to prove that
					\begin{align*}
						-\frac{k^2}{2(k+t)} - \frac{k}{2(k+t)} + k + 1 \geq \frac{k+t}{2} + \frac{1}{2},
					\end{align*}
				which can be further simplified to $t\geq t^2$. Since $0\leq t < 1$, we conclude the proof. When $t = 0$, thus $\left\lfloor\frac{1}{p(x_1)}\right\rfloor = \frac{1}{p(x_1)}$, the equality is achieved.
			\end{proof}
		Again, we derive the conditional counterpart of Eq.(\ref{eq_lem_gp2}).
			\vspace*{12pt}
			\begin{mycoro}
				\label{coro_gp2}
				Let $\mathrm{G}(X|Y)$ and $\mathrm{P}_{err}(X|Y)$ be defined in Eq.(\ref{eq_def_cond_guesswork}) and Eq.(\ref{eq_def_cond_p}), respectively. It holds that
					\begin{align}
						\label{eq_coro_gp2}
						\mathrm{G}(X|Y) \geq \frac{1}{2(1- \mathrm{P}_{err}(X|Y) )} + \frac{1}{2}.
					\end{align}
			\end{mycoro}
			\begin{proof}
				We have the following inference:
					\begin{align*}
						\mathrm{G}(X|Y) & = \sum_{j=1}^m p(y_j) \mathrm{G}(X|Y=y_j) \\
										& \geq \sum_{j=1}^m p(y_j) (\frac{1}{2( 1 - \mathrm{P}_{err}(X|Y=y_j) )} + \frac{1}{2}) \\
										& \geq \frac{1}{2( 1 - \sum_{j=1}^m p(y_j)\mathrm{P}_{err}(X|Y=y_j) )} + \frac{1}{2} \\
										& = \frac{1}{2( 1 - \mathrm{P}_{err}(X|Y) )} + \frac{1}{2},
					\end{align*}
				where the second inequality is from Jensen's inequality.
			\end{proof}

		We observe that there is an even stronger connection between guesswork and error probability for the special case where $n=2$.
			\vspace*{12pt}
			\begin{mylemma}
				\label{lem_gp3}
				If the alphabet of $X$ comprises exactly two elements, i.e., $n=2$, then it holds that
					\begin{itemize}
						\item[(i)]
							$	\mathrm{G}(X) = \mathrm{P}_{err}(X) + 1$;
						\item[(ii)]
							$	\mathrm{G}(X|Y) = \mathrm{P}_{err}(X|Y) + 1$.
					\end{itemize}
			\end{mylemma}
			\vspace*{12pt}
			
		We omit the proof of this lemma, because it is easy from the definitions of guesswork and error probability. Based on the preceding results, we obtain the relation between $\mathrm{G}^{opt}(\mathcal{E})$ and $\mathrm{P}^{opt}_{err}(\mathcal{E})$ as follows.
			\vspace*{12pt}
			\begin{mythm}
				\label{thm_gp}
				Let $\mathcal{E}$ be a quantum encoding of a random variable $X$, and $\mathrm{G}^{opt}(\mathcal{E})$ and $\mathrm{P}^{opt}_{err}(\mathcal{E})$ be defined in Eq.(\ref{eq_def_min_guesswork}) and Eq.(\ref{eq_def_min_p}), respectively. It holds that
						\begin{align}
							\label{eq_thm_gp}
							\frac{1}{2(1-\mathrm{P}^{opt}_{err}(\mathcal{E}))} + \frac{1}{2} \leq \mathrm{G}^{opt}(\mathcal{E}) \leq \frac{n}{2}\mathrm{P}^{opt}_{err}(\mathcal{E}) + 1
						\end{align}
				and if $n=2$, 
							\begin{align}
								\label{eq_thm_gp2}
								\mathrm{G}^{opt}(\mathcal{E}) = \mathrm{P}^{opt}_{err}(\mathcal{E}) + 1.
							\end{align}
			\end{mythm}
			\begin{proof}
				We prove the left inequality in Eq.(\ref{eq_thm_gp}). Let $\Pi$ be an optimal POVM achieving $\mathrm{G}^{opt}(\mathcal{E})$ and $Y$ the corresponding random variable induced by $\mathcal{E}$ and $\Pi$. Applying Eq.(\ref{eq_coro_gp2}), we have the following inference:
					\begin{align*}
						\mathrm{G}^{opt}(\mathcal{E}) & = \mathrm{G}(X|Y) \\
												& \geq \frac{1}{2(1-\mathrm{P}_{err}(X|Y))} + \frac{1}{2} \\
												& \geq \frac{1}{2(1-\mathrm{P}^{opt}_{err}(\mathcal{E}))} + \frac{1}{2}.
					\end{align*}
				Eq.(\ref{eq_thm_gp2}) and the right inequality in Eq.(\ref{eq_thm_gp}) follows from similar reasoning but using optimal POVM achieving $\mathrm{P}_{err}^{opt}(\mathcal{E})$ instead.
			\end{proof}
		This theorem states that the minimum error criterion coincides with the minimum guesswork criterion in the context of two state discrimination. However, regarding to general cases, the two criteria may not agree with each other, though there still exists a weak correlation between them as shown in Eq.(\ref{eq_thm_gp}).
		The following example shows their disagreement when three states are presented.
			\begin{myexa}
				\label{exa_trine}
				We consider the so-called \emph{trine ensemble}~\cite{HausladenW1994,Peres1990}, which consists of three equiprobable pure qubit states (living in a $2$-dimensional Hilbert space), given by
					\begin{align*}
						|\phi_{x_1}\rangle &= |0\rangle, \\
						|\phi_{x_2}\rangle &= - \frac{1}{2} |0\rangle + \frac{\sqrt{3}}{2} |1\rangle, \\
						|\phi_{x_3}\rangle &= - \frac{1}{2} |0\rangle - \frac{\sqrt{3}}{2} |1\rangle.
					\end{align*}
				For the trine ensemble, the $3$-component POVM, denoted by $\Pi^E$, with each operator given by
					\begin{align*}
						\pi_{y_i} = \frac{2}{3} \rho_{x_i} = \frac{2}{3} |\phi_{x_i}\rangle\langle\phi_{x_i}|,
					\end{align*}
				achieves the minimum error with $\mathrm{P}_{err}(X|\Pi^{E}) = 1/3$. This POVM is known as the \emph{square-root measurement}~\cite{HausladenW1994}. On the other hand, in Section~\ref{sec_gu} we will show that the POVM $\Pi^G = \{ 2/3 |\psi_{x_k}\rangle\langle\psi_{x_k}|: 1\leq k\leq 3 \}$ with
					\begin{align*}
						|\psi_{x_k}\rangle = \cos(\frac{2k\pi}{3}-\frac{7\pi}{12})|0\rangle + \sin(\frac{2k\pi}{3} - \frac{7\pi}{12}) |1\rangle
					\end{align*}
				achieves the minimum guesswork with $\mathrm{G}(X|\Pi^{G}) = 2 - \sqrt{3}/3$. It is also easy to verify that 
					\begin{align*}
						\mathrm{P}_{err}(X|\Pi^G) &= \frac{2}{3} - \frac{\sqrt{3}}{6} > \mathrm{P}_{err}(X|\Pi^E), \\
						\mathrm{G}(X|\Pi^E) &= \frac{3}{2} > \mathrm{G}(X|\Pi^G),
					\end{align*}
				where we write $\Pi^E$ or $\Pi^G$ for $Y$ to avoid confusion. Indeed, by observing the proof of the optimality of $\Pi^G$ in Section \ref{sec_gu}, we can conclude that for the trine ensemble there does not exist any POVM which can achieve both minimum error and minimum guesswork.
			\end{myexa}


	\section{Upper and lower information-theoretic bounds on minimum guesswork}
	\noindent
		\label{sec_bound}
		As a matter of fact, the optimization problem in MED is usually hard to solve analytically. Closed-form result or optimal measurement is only known for some special quantum systems, e.g., the case with exactly two states~\cite{Helstrom1976}, equiprobable symmetric states~\cite{BanKMH1997}, or multiply symmetric states~\cite{Barnett2001}. We believe that it is also the case for MGD, because of the analogy between these two problems. In light of this, upper/lower bounds on minimum guesswork are as desirable as those on minimum error~\cite{Barnum2002,Montanaro2007,Montanaro2008,Qiu2008,QiuL2010}. In what follows, we start by reviewing some existing upper/lower bounds on guesswork in classic setting. We then combine them with the celebrated \emph{Holevo bound} and the less well-known \emph{subentropy bound} on accessible information, resulting in upper/lower bounds on minimum guesswork in the quantum setting.

		In~\cite{Massey1994}, Massey proved the following lower bound on guesswork $\mathrm{G}(X)$.
			\vspace*{12pt}
			\begin{mylemma}[\cite{Massey1994}] 
				\label{lem_lower_bound}
				Let $\mathrm{G}(X)$ be defined in Eq.(\ref{eq_def_guesswork}). Provided $\mathrm{H}(X) \geq 2$, it holds that
					\begin{align*}
						\mathrm{G}(X) \geq \frac{1}{4} \cdot 2^{\mathrm{H}(X)} + 1.
					\end{align*}
			\end{mylemma}
			\vspace*{12pt}
		Note that $\mathrm{H}(X) \triangleq - \sum_{i=1}^n p(x_i)\log p(x_i)$ is the Shannon entropy. The logarithm is taken with base $2$. Let us also recall the conditional Shannon entropy:
			\begin{align*}
				\mathrm{H}(X|Y) & \triangleq \sum_{j=1}^m p(y_j)\mathrm{H}(X|Y = y_j), \\
								& = - \sum_{j=1}^m p(y_j) \sum_{i=1}^n p(x_i|y_j) \log p(x_i|y_j).
			\end{align*}
		Based on Massey's result, we derive a similar lower bound on conditional guesswork.
			\vspace*{12pt}
			\begin{mycoro}
				\label{coro_lower_bound}
				Let $\mathrm{G}(X|Y)$ be defined in Eq.(\ref{eq_def_cond_guesswork}). Provided $\mathrm{H}(X|Y = y_j) \geq 2$ for each $y_j$, it holds that
					\begin{align}
						\label{eq_coro_lower_bound}
						\mathrm{G}(X|Y) \geq \frac{1}{4} \cdot 2^{\mathrm{H}(X|Y)} + 1.
					\end{align}
			\end{mycoro}
				\begin{proof}
					\label{proof_coro_lower_bound}
					We have the following inference:
						\begin{align*}
							\mathrm{G}(X|Y) & = \sum_{j=1}^{m} p(y_j) \mathrm{G}(X|Y = y_j) \\
									& \geq \sum_{j=1}^m p(y_j) (\frac{1}{4}\cdot 2^{\mathrm{H}(X|Y = y_j)} + 1) \\
									& \geq \frac{1}{4} \cdot 2^{\sum_{j=1}^m p(y_j) \mathrm{H}(X|Y=y_j)} + 1 \\
									& = \frac{1}{4} \cdot 2^{\mathrm{H}(X|Y)} + 1,
						\end{align*}
						where the second inequality is from Jensen's inequality.
				\end{proof}

		Upper bound on guesswork in terms of the Shannon entropy also exists.
			\vspace*{12pt}
			\begin{mylemma}[\cite{McelieceY1995}]
				\label{lem_upper_bound}
				Let $\mathrm{G}(X)$ be defined in Eq.(\ref{eq_def_guesswork}). It holds that
					\begin{align}
						\label{eq_lem_upper_bound}
						\mathrm{G}(X) \leq \frac{n-1}{2\log n} \mathrm{H}(X) + 1.
					\end{align}
			\end{mylemma}
			\vspace*{12pt}
		Again, we present the conditional counterpart of Eq.(\ref{eq_lem_upper_bound}).
			\vspace*{12pt}
			\begin{mycoro}
				\label{coro_upper_bound}
				Let $\mathrm{G}(X|Y)$ be defined in Eq.(\ref{eq_def_cond_guesswork}). It holds that
					\begin{align}
						\label{eq_coro_upper_bound}
						\mathrm{G}(X|Y) \leq \frac{n-1}{2\log n} \mathrm{H}(X|Y) + 1.
					\end{align}
			\end{mycoro}
			\begin{proof}
				\label{proof_coro_upper_bound}
				\begin{align*}
					\mathrm{G}(X|Y) & = \sum_{j=1}^{m} p(y_j) \mathrm{G}(X|Y = y_j) \\
									 & \leq \sum_{j=1}^m p(y_j) (\frac{n-1}{2\log n}\mathrm{H}(X|Y=y_j) + 1) \\
									 & = \frac{n-1}{2\log n} \mathrm{H}(X|Y) + 1.
				\end{align*}
			\end{proof}

		Now, let us review two bounds on accessible information.
		Let $\mathcal{E}$ be a quantum encoding of a random variable $X$, and $Y$ be the random variable induced by $\mathcal{E}$ and some POVM $\Pi$ as described in Section~\ref{sec_qgp}. The \emph{accessible information} $\mathrm{I}_{acc}(\mathcal{E})$ of the ensemble $\mathcal{E}$ is defined as the maximum mutual information between $X$ and $Y$ obtainable via varying the POVM $\Pi$:
			\begin{align*}
				\mathrm{I}_{acc}(\mathcal{E}) \triangleq \max_{\Pi \in\mathcal{M}} \mathrm{I}(X:Y),
			\end{align*}
		where $\mathrm{I}(X:Y) = \mathrm{H}(X) - \mathrm{H}(X|Y)$.
		Although accessible information is difficult to characterize analytically, various upper and lower bounds have been found. In a celebrated paper~\cite{Holevo1973b}, Holevo bounded $\mathrm{I}_{acc}(\mathcal{E})$ as follows:
			\begin{align}
				\label{eq_holevo_bound}
				\mathrm{I}_{acc}(\mathcal{E}) \leq \chi(\mathcal{E}) \triangleq \mathrm{S}(\sum_{i=1}^n p(x_i)\rho_{x_i}) - \sum_{i=1}^n p(x_i) \mathrm{S}(\rho_{x_i}),
			\end{align}
		where the \emph{von Neumann entropy} $\mathrm{S}(\rho)$ of quantum state $\rho$ is given by $\mathrm{S}(\rho) \triangleq - \mathrm{Tr}(\rho\log\rho)$ which is a natural extension of the Shannon entropy. The quantity $\chi(\mathcal{E})$ is usually referred to as the \emph{Holevo information} of the ensemble $\mathcal{E}$. 
		Interestingly, accessible information can also be lower bounded by a quantity which has an analogous form to the Holevo information. With the notion of \emph{subentropy} which is defined by
			\begin{align*}
				\mathrm{Q}(\rho) \triangleq -\sum_{k}\prod_{l\not=k}\frac{\lambda_k}{\lambda_k - \lambda_l} \lambda_k \log \lambda_k,
			\end{align*}
		with $\lambda_k$'s being the eigenvalues of the state $\rho$, Jozsa, Robb and Wootters~\cite{JozsaRW1994} obtained the following inequality:
			\begin{align}
				\label{eq_jrw_bound}
				\mathrm{I}_{acc}(\mathcal{E}) \geq \Lambda(\mathcal{E}) \triangleq \mathrm{Q}(\sum_{i=1}^n p(x_i)\rho_{x_i}) - \sum_{i=1}^n p(x_i) \mathrm{Q}(\rho_{x_i}).
			\end{align}

		Combining Eq.(\ref{eq_coro_lower_bound}) and Eq.(\ref{eq_holevo_bound}), we obtain the following lower bound on minimum guesswork.
			\vspace*{12pt}
			\begin{mythm}
				\label{thm_lb}
				Let $\mathcal{E}$ be a quantum encoding of a random variable $X$ and $\mathrm{G}^{opt}(\mathcal{E})$ be defined in Eq.(\ref{eq_def_min_guesswork}). Provided $\mathrm{H}(X_{\pi}) \geq 2$ for any $\pi\in\mathcal{P}$, it holds that
					\begin{align}
						\label{eq_thm_lb}
						\mathrm{G}^{opt}(\mathcal{E}) \geq \frac{1}{4} \cdot 2^{ \mathrm{H}(X)-\chi(\mathcal{E}) } + 1.
					\end{align}
			\end{mythm}
			\begin{proof}
				\label{proof_thm_lb}
				Since
					\begin{align*}
						\mathrm{I}_{acc}(\mathcal{E}) = \mathrm{H}(X) - \min_{\Pi\in\mathcal{M}} \mathrm{H}(X|Y) \leq \chi(\mathcal{E}),
					\end{align*}
				we have that
					\begin{align}
						\label{eq_proof_lb1}
						\min_{\Pi\in\mathcal{M}} \mathrm{H}(X|Y) \geq \mathrm{H}(X) - \chi(\mathcal{E}).
					\end{align}
				Let $\Pi$ be an optimal POVM achieving $\mathrm{G}^{opt}(\mathcal{E})$. Due to the guarantee that $\mathrm{H}(X_{\pi}) \geq 2$ for any $\pi\in\mathcal{P}$, we are safe to apply Eq.(\ref{eq_coro_lower_bound}):
					\begin{align*}
						\mathrm{G}(X|Y) \geq \frac{1}{4}\cdot 2^{\mathrm{H}(X|Y)} + 1,
					\end{align*}
				where $Y$ is the random variable induced by $\mathcal{E}$ and $\Pi$. Then, Eq.(\ref{eq_proof_lb1}) and Eq.(\ref{eq_coro_lower_bound}) together imply that
					\begin{align*}
						\mathrm{G}^{opt}(\mathcal{E}) \geq \frac{1}{4}\cdot 2^{\mathrm{H}(X)-\chi(\mathcal{E})} + 1,
					\end{align*}
				as required.
			\end{proof}
		From the above proof, we see that the equality in Eq.(\ref{eq_thm_lb}) holds if and only if there exists an optimal POVM $\Pi$ achieving $\mathrm{G}^{opt}(\mathcal{E})$ such that $\mathrm{G}(X|Y) = \frac{1}{4}\cdot 2^{\mathrm{H}(X|Y)} + 1$ and $\mathrm{H}(X|Y) = \mathrm{H}(X) - \chi(\mathcal{E})$. To satisfy the first equation, we have to require $\mathrm{G}(X_{\pi}) = \frac{1}{4}\cdot 2^{\mathrm{H}(X_{\pi})} + 1$ for each $\pi\in\Pi$ (see the proof of Corollary~\ref{coro_lower_bound}). Consequently, the alphabet of $X$ must be countably infinite and each $X_{\pi}$ must obey the geometric distribution $\{\frac{1}{2},\frac{1}{2^2},\cdots\}$~\cite{Massey1994}. A trivial case is when $X_{\pi}$ are identical, e.g.,
					\begin{align*}
						\mathrm{Pr}(X_{\pi} = x_i) = \frac{1}{2^i}, \qquad i = 1,2,\cdots,
					\end{align*}
		for each $\pi\in\Pi$. In this case, we have that $\mathrm{H}(X) = \mathrm{H}(X|Y) = \mathrm{G}(X|Y) = 2$, which further implies that $\chi(\mathcal{E}) = 0$ as required by the second equation. Therefore, Alice's choice of $\mathcal{E}$ has to be a trivial encoding with all the states being identical! Indeed, this special system satisfies the ``no-measurement'' condition discussed in Section~\ref{sec_n_m}. That is, Bob cannot decrease his prior guesswork $\mathrm{G}(X)$ by applying any measurement.

		Nonetheless, there also exists non-trivial quantum encoding $\mathcal{E}$ satisfying the equality condition of Eq.(\ref{eq_thm_lb}). We give an example below. Let $X$ be a random variable with the associated distribution given by
			\begin{align*}
				p(x_1) &= p(x_2)=\frac{3}{8},\\
				p(x_i) &= \frac{1}{2^i}, \qquad i = 3,4,\cdots.
			\end{align*}
		Consider the following quantum encoding $\mathcal{E}$ of $X$:
			\begin{align*}
				\rho_{x_1} &= \frac{2}{3} |0\rangle\langle 0| + \frac{1}{3} |1\rangle\langle 1|, \\
				\rho_{x_2} &= \frac{1}{3} |0\rangle\langle 0| + \frac{2}{3} |1\rangle\langle 1|, \\
				\rho_{x_i} &= \frac{1}{2} |0\rangle\langle 0| + \frac{1}{2} |1\rangle\langle 1|, \qquad i = 3,4,\cdots.
			\end{align*}
		It is straightforward to verify that $\mathrm{H}(X_{\pi}) \geq 2$ for any $\pi\in\mathcal{P}$. Then, with the POVM $\Pi = \{\pi_{y_1} = |0\rangle\langle 0|, \pi_{y_2} = |1\rangle\langle 1|\}$, we see that both $X_{\pi_{y_1}}$ and $X_{\pi_{y_2}}$ obey the geometric distribution $\{\frac{1}{2},\frac{1}{2^2},\cdots\}$, and thus $\mathrm{H}(X|Y) = \mathrm{G}(X|Y) = 2$. On the other hand, we can calculate that $\mathrm{H}(X) = 13/4 - 3\log 3/4$ and $\chi(\mathcal{E}) = 5/4 - 3\log 3/4$. Hence, it holds that $\mathrm{H}(X|Y) = \mathrm{H}(X) - \chi(\mathcal{E})$ as required by the equality condition of Eq.(\ref{eq_thm_lb}).

			\vspace*{12pt}
			\begin{myrmk}
				\label{rmk_2}
				We note that in order to apply Eq.(\ref{eq_thm_lb}) we need to verify in advance the precondition $\mathrm{H}(X_{\pi}) \geq 2$ for any non-zero PSD $\pi$. This constraint, which limits the applicability of the bound, originates from the one in Lemma~\ref{lem_lower_bound}. Intuitively, it can only be fulfilled by quantum states whose spanning spaces overlap to a certain extent. We give an example where this constraint holds. Consider a $5$-dimensional Hilbert space $\mathcal{H}_5$ with an orthogonal basis $\{|1\rangle,\cdots,|5\rangle\}$. Let $\mathcal{E} = \{(1/5,\rho_{x_i}):1\leq i\leq 5\}$ be a quantum encoding of some random variable $X$, with $\rho_{x_i}$ defined by
					\begin{align*}
						\rho_{x_i} = \frac{1}{4}(I - |i\rangle\langle i|),
					\end{align*}
				where $I$ is the identity operator on $\mathcal{H}_5$. For any $\pi\in\mathcal{P}$, we can calculate that $\mathrm{H}(X_{\pi}) = - \sum_{i=1}^5 p_i\log p_i$ with
					\begin{align*}
						p_i = \frac{1}{4}(1-\frac{\langle i|\pi|i\rangle}{\mathrm{Tr}(\pi)}).
					\end{align*}
				It is easy to verify that $\mathrm{H}(X_{\pi}) \geq 2$.
			\end{myrmk}
			\vspace*{12pt}

		To bound minimum guesswork in the other direction, we combine Eq.(\ref{eq_coro_upper_bound}) and Eq.(\ref{eq_jrw_bound}).
			\vspace*{12pt}
			\begin{mythm}
				\label{thm_ub}
				Let $\mathcal{E}$ be a quantum encoding of a random variable $X$ and $\mathrm{G}^{opt}(\mathcal{E})$ be defined in Eq.(\ref{eq_def_min_guesswork}). It holds that
					\begin{align}
						\label{eq_thm_ub}
						\mathrm{G}^{opt}(\mathcal{E}) \leq \frac{n-1}{2\log n} (\mathrm{H}(X) - \Lambda(\mathcal{E})) + 1.
					\end{align}
			\end{mythm}
			\begin{proof}
				\label{proof_thm_ub}
				Since
					\begin{align*}
						\mathrm{I}_{acc}(\mathcal{E}) = \mathrm{H}(X) - \min_{\Pi\in\mathcal{M}} \mathrm{H}(X|Y) \geq \Lambda(\mathcal{E}), 
					\end{align*}
				we have that
					\begin{align*}
						\min_{\Pi\in\mathcal{M}} \mathrm{H}(X|Y) \leq \mathrm{H}(X) - \Lambda(\mathcal{E}).
					\end{align*}
				Consequently, there must exist some $\Pi$ and $Y$ such that $\mathrm{H}(X|Y) \leq \mathrm{H}(X) - \Lambda(\mathcal{E})$. Applying Eq.(\ref{eq_coro_upper_bound}), we obtain the following inference:
					\begin{equation}
						\label{eq_proof_ub}
						\begin{split}
						\mathrm{G}^{opt}(\mathcal{E}) & \leq \mathrm{G}(X|Y) \\
											& \leq \frac{n-1}{2\log n} \mathrm{H}(X|Y) + 1 \\
											& \leq \frac{n-1}{2\log n} (\mathrm{H}(X) - \Lambda(\mathcal{E})) + 1.
						\end{split}
					\end{equation}
			\end{proof}
		Let us examine when the inequality in Eq.(\ref{eq_thm_ub}) is saturated. First, it is required that $\mathrm{I}_{acc}(\mathcal{E}) = \Lambda(\mathcal{E})$. Otherwise, there must exist a random variable $Y$ induced by $\mathcal{E}$ and some POVM $\Pi$ such that $\mathrm{H}(X|Y) < \mathrm{H}(X) - \Lambda(\mathcal{E})$, which further implies that $\mathrm{G}^{opt}(\mathcal{E})$ is strictly less than the RHS term in Eq.(\ref{eq_thm_ub}). According to the discussion in~\cite{JozsaRW1994}, the lower bound $\Lambda(\mathcal{E})$ on the accessible information of $\mathcal{E}$ can only be achieved by the so-called \emph{Scrooge ensemble} or trivially by an ensemble with all the states being identical. It was showed that the amount of information that we can obtain from the Scrooge ensemble by performing some complete POVM is independent of the choice of POVM.
		
		Second, for any complete POVM $\Pi$ and the corresponding random variable $Y$ induced by $\mathcal{E}$ and $\Pi$, we require that $\mathrm{G}(X|Y) = \frac{n-1}{2\log n}\mathrm{H}(X|Y) + 1$ to saturate the second inequality in Eq.(\ref{eq_proof_ub}). It then follows from the proof of Corollary~\ref{coro_upper_bound} that $\mathrm{G}(X_{\pi}) = \frac{n-1}{2\log n} \mathrm{H}(X_{\pi}) + 1$ for any $\pi\in\mathcal{P}$. On the other hand, the equality in Eq.(\ref{eq_lem_upper_bound}) holds if and only if $X$ obeys either the uniform distribution, i.e., $p(x_i)= 1/n$ for each $1\leq i\leq n$, or trivially a point distribution, i.e., $p(x_i) = 1$ for some $1\leq i\leq n$~\cite{McelieceY1995}. As a consequence, all the quantum states in $\mathcal{E}$ must be either mutually orthogonal or identical. Since mutually orthogonal states cannot form a Scrooge ensemble, we conclude that the equality in Eq.(\ref{eq_thm_ub}) holds if and only if all the states in $\mathcal{E}$ are identical.
		

	\section{Geometrically uniform states}
	\noindent
		\label{sec_gu}
		In this section, we consider MGD for a special type of symmetric quantum states, the \emph{geometrically uniform} states~\cite{EldarF2001}. We provide sufficient conditions when some geometrically uniform measurement achieves minimum guesswork. With this technique, we are able to prove the optimality of the POVM $\Pi^G$, as given in Example~\ref{exa_trine}, for the trine ensemble.
		
		Let $\mathcal{G} = \{U_i: 1\leq i\leq n\}$ be a finite group which contains $n$ unitary operators. The \emph{identity} element of $\mathcal{G}$ and the \emph{inverse} of $U_i$ is denoted by $I$ and $U_i^{\dagger}$, respectively. 
		An ensemble of states $\mathcal{E} = \{(p(x_i),\rho_{x_i}):1\leq i\leq n\}$ is called \emph{geometrically uniform} if $\rho_{x_i}$'s can be generated by a state $\rho_0$ and a group $\mathcal{G}$ such that $\rho_{x_i} = U_i \rho_0 U_i^{\dagger}$. It is also assumed that $p(x_i) = 1/n$, as required by symmetry. Similarly, we can define such symmetry for quantum measurements. A POVM $\Pi = \{\pi_{y_i}:1\leq i\leq n\}$ is called geometrically uniform if the measurement operators are generated by some $\pi_0$ and $\mathcal{G}$ such that $\pi_{y_i} = U_i \pi_0 U_i^{\dagger}$. In~\cite{EldarMV2004}, Eldar, Megretski, and Verghese showed that the square-root measurement of a geometrically uniform ensemble is also geometrically uniform, and that there always exists a geometrically uniform measurement which achieves the minimum error for this ensemble.
		
		We state our sufficient conditions as follows. Given a unitary operator $V$ and a POVM $\Pi = \{\pi_{y_i}:1\leq i\leq n\}$, we use $V\Pi V^{\dagger}$ to denote the new POVM $\{V\pi_{y_i} V^{\dagger}:1\leq i\leq n\}$. 
			\vspace*{12pt}
			\begin{mythm}
				\label{thm_geo}
				Let $\mathcal{E} = \{(p(x_i),\rho_{x_i}):1\leq i\leq n\}$ be a quantum encoding of a random variable $X$ such that $\mathcal{E}$ is geometrically uniform being generated by a state $\rho_0$ and a group $\mathcal{G} = \{U_i: 1\leq i\leq n\}$. Given a geometrically uniform POVM $\Pi = \{\pi_{y_i}: 1\leq i\leq n\}$, which is generated by $\pi_0$ and $\mathcal{G}$, if there exists a unitary operator $V$ such that
					\begin{itemize}
						\item[(i)]
						$V U_i = U_i V$, for each $1\leq i\leq n$, and
						\item[(ii)]
						$\mathrm{G}(X_{V\pi_0 V^{\dagger}}) = \min_{\pi\in\mathcal{P}_1} \mathrm{G}(X_{\pi})$,
					\end{itemize}
				then $V\Pi V^{\dagger}$ achieves the minimum guesswork for $\mathcal{E}$, i.e., $\mathrm{G}^{opt}(\mathcal{E}) = \mathrm{G}(X|V\Pi V^{\dagger})$.
			\end{mythm}
			\begin{proof}
				The proof of this theorem consists of two steps. First, we show that the minimum of $\mathrm{G}(X_\pi)$ over all $\pi \in \mathcal{P}$ can always be obtained in $\mathcal{P}_1$; that is, $\min_{\pi\in \mathcal{P}} \mathrm{G}(X_\pi) = \min_{\pi\in \mathcal{P}_1} \mathrm{G}(X_\pi)$. To see this, consider Eq.(\ref{eq_proof_thm_cm}) used in the proof of Proposition~\ref{thm_cm}. For any $\pi_{y_1},\pi_{y_2},\pi_{y_3}\in \mathcal{P}$ such that $\pi_{y_1} = \pi_{y_2} + \pi_{y_3}$, we have proved that
					\begin{align*}
						p(y_1)\mathrm{G}(X_{\pi_{y_1}}) \geq p(y_2)\mathrm{G}(X_{\pi_{y_2}}) + p(y_3)\mathrm{G}(X_{\pi_{y_3}}),
					\end{align*}
				where $p(y_j) = \sum_{i=1}^n p(x_i)\mathrm{Tr}(\rho_{x_i}\pi_{y_j})$. Since $p(y_1) = p(y_2)+p(y_3)$, it must hold either $\mathrm{G}(X_{\pi_{y_1}}) \geq \mathrm{G}(X_{\pi_{y_2}})$ or $\mathrm{G}(X_{\pi_{y_1}}) \geq \mathrm{G}(X_{\pi_{y_3}})$. Consequently, only rank-one PSD operators need to be considered in minimizing $\mathrm{G}(X_{\pi})$.

				Second, recall that $\mathrm{G}(X|Y)$ can be written as $\sum_{j=1}^n p(y_j)\mathrm{G}(X_{\pi_{y_j}})$. Provided the conditions (i) and (ii) are satisfied, we prove the optimality of the POVM $V\Pi V^{\dagger}$ by showing that $\mathrm{G}(X_{V\pi_{y_j}V^{\dagger}}) = \mathrm{G}(X_{V\pi_{0}V^{\dagger}})$ for each $1\leq j\leq n$. As the identity matrix $I$ is contained in $\mathcal{G}$, without loss of generality we assume that $U_1 = I$ and thus $\pi_{y_1} = \pi_0$. It holds that
					\begin{align*}
						\mathrm{G}(X_{V\pi_{y_j}V^{\dagger}}) & = \frac{ \sum_{i=1}^n \sigma_{j}(i) \mathrm{Tr}(\rho_{x_i}V\pi_{y_j}V^{\dagger}) }{ \sum_{i=1}^n \mathrm{Tr}(\rho_{x_i}V\pi_{y_j}V^{\dagger}) } \\
							& = \frac{ \sum_{i=1}^n \sigma_{j}(i) \mathrm{Tr}(U_i\rho_0 U_i^{\dagger}VU_j\pi_0 U_j^{\dagger}V^{\dagger}) }{ \sum_{i=1}^n \mathrm{Tr}(U_i\rho_0 U_i^{\dagger}VU_j\pi_0 U_j^{\dagger}V^{\dagger}) }\\
							& = \frac{ \sum_{i=1}^n \sigma_{j}(i) \mathrm{Tr}(U_j^{\dagger}U_i\rho_0 U_i^{\dagger}U_jV \pi_0 V^{\dagger}) }{ \sum_{i=1}^n \mathrm{Tr}(U_j^{\dagger}U_i\rho_0 U_i^{\dagger}U_j V\pi_0 V^{\dagger}) }\\
							& = \frac{ \sum_{i=1}^n \sigma_{1}(i) \mathrm{Tr}(U_i\rho_0 U_i^{\dagger}V\pi_{y_1} V^{\dagger}) }{ \sum_{i=1}^n \mathrm{Tr}(U_i\rho_0 U_i^{\dagger}V\pi_{y_1} V^{\dagger}) }\\
							& = \mathrm{G}(X_{V\pi_{y_1} V^{\dagger}}) = \mathrm{G}(X_{V\pi_0 V^{\dagger}}),
					\end{align*}
				where $\sigma_j$ is an optimal permutation with respect to $X_{V\pi_{y_j}V^{\dagger}}$ for $1\leq j\leq n$. In this reasoning, the third equation follows from the condition (i) and the fourth is due to that $\{U_j^{\dagger}U_i:1\leq i\leq n\} = \mathcal{G}$ for each $1\leq j\leq n$. Then, the condition (ii) implies the optimality of $V\Pi V^{\dagger}$.
			\end{proof}
		
			It is worth noting that the above theorem just reduces the general minimization problem into another restricted one, i.e., the term $\min_{\pi\in\mathcal{P}_1} \mathrm{G}(X_{\pi})$ in the condition (ii), which in itself is not easy to solve.
			
			Nevertheless, this theorem suffices to prove that the POVM $\Pi^{G}$ in Example~\ref{exa_trine} is optimal for the trine ensemble. It can be easily verified that the trine ensemble and the POVM $\Pi^E$ are geometrically uniform. The three states of the ensemble can be obtained from one another by a rotation $R_y(4\pi/3)$ about $y$-axis of the Bloch sphere, with $R_y(\theta)$ defined as
				\begin{align*}
					R_y(\theta) = \begin{bmatrix}
									\cos\frac{\theta}{2} & - \sin\frac{\theta}{2} \\
									\sin\frac{\theta}{2} & \cos\frac{\theta}{2}
								  \end{bmatrix},
				\end{align*}
			where $0\leq \theta <2\pi$. The POVM $\Pi^G$ can be obtained from $\Pi^E$ in a similar way: $\Pi^G = R_y(\pi/6)\Pi^E R_y(\pi/6)^{\dagger}$. Since rotations about the same axis commute, i.e., $R_y(\alpha)R_y(\beta) = R_y(\beta)R_y(\alpha)$, the condition (i) of Theorem~\ref{thm_geo} is satisfied. It remains to verify that the measurement operator $\frac{2}{3}R_y(\pi/6) |0\rangle\langle 0| R_y(\pi/6)^{\dagger}$ actually achieves $\min_{\pi\in\mathcal{P}_1} \mathrm{G}(X_{\pi})$. Let $\pi' = |w\rangle\langle w|$ be an arbitrary rank-one PSD operator on a $2$-dimensional Hilbert space, $|w\rangle$ can be given as
				\begin{align*}
					|w\rangle = \cos\alpha|0\rangle + e^{i\beta}\sin\alpha|1\rangle,
				\end{align*}
			where $0\leq \alpha,\beta < 2\pi$. To calculate the value of $\mathrm{G}(X_{\pi'})$, we need the following probabilities:
				\begin{align*}
					\mathrm{Pr}(X_{\pi'} = x_1) & = \frac{1}{3} (1+ \cos2\alpha),\\
					\mathrm{Pr}(X_{\pi'} = x_2) & = \frac{1}{6} (2- \cos2\alpha - \sqrt{3}\sin2\alpha\cos\beta), \\
					\mathrm{Pr}(X_{\pi'} = x_3) & = \frac{1}{6} (2- \cos2\alpha + \sqrt{3}\sin2\alpha\cos\beta).
				\end{align*}
			There are two cases of the value of $\mathrm{G}(X_{\pi'})$:
				\begin{align*}
					 \begin{cases}
												2 - \frac{\sqrt{3}}{3} \sin2\alpha\cos\beta, \qquad \mbox{ if } \sqrt{3}\cos2\alpha  \leq \sin2\alpha\cos\beta, \\
												2 - \frac{1}{2}\cos2\alpha - \frac{\sqrt{3}}{6}\sin2\alpha\cos\beta, \qquad  \mbox{ otherwise.}
											\end{cases}
				\end{align*}
			For the first case, $\mathrm{G}(X_{\pi'})$ achieves the minimum value $2 - \sqrt{3}/3$ with $\alpha = \pi/4$ and $\beta = 0$. For the second case, $\mathrm{G}(X_{\pi'})$ achieves the same minimum value with $\alpha = \pi/12$ and $\beta = 0$. Since $\pi' = R_y(\pi/6) |0\rangle\langle 0| R_y(\pi/6)^{\dagger}$ when $\alpha = \pi/12$ and $\beta = 0$, it follows that $\Pi^G$ is optimal for the trine ensemble.

	\section{When does making measurement provide no benefit?}
	\noindent
		\label{sec_n_m}
		So far, we have focused on the adversarial viewpoint taken by Bob aiming to minimize his guesswork. In this section, let us take the protective viewpoint of Alice, whose task is to choose an optimal encoding ensemble to maximize Bob's minimum guesswork. Similar max-min problem has been addressed in the context of MED~\cite{Hunter2003}. Formally, we are concerned with the following goal:
			\begin{align*}
				\max_{\mathcal{E}} \mathrm{G}^{opt}(\mathcal{E}).
			\end{align*}
		Since it always holds that $\mathrm{G}(X|Y) \leq \mathrm{G}(X)$~\cite{Arikan1996}, we simply know that
			\begin{align*}
				\max_{\mathcal{E}} \mathrm{G}^{opt}(\mathcal{E}) = \max_{\mathcal{E}} \min_{\Pi} \mathrm{G}(X|Y) \leq \mathrm{G}(X).
			\end{align*}
		In what follows, we show that for any classical information source $X$, Alice can always find a quantum encoding $\mathcal{E}$ which renders Bob's measurement useless in the sense that no measurement can reduce his prior guesswork, i.e., $\max_{\mathcal{E}}\mathrm{G}^{opt}(\mathcal{E}) = \mathrm{G}(X)$.

		We start from the proof of $\mathrm{G}(X|Y) \leq \mathrm{G}(X)$:
			\begin{align*}
				\mathrm{G}(X|Y) & = \sum_{j=1}^m p(y_j) \sum_{i=1}^n \sigma_j(i) p(x_i|y_j) \\
								& \leq \sum_{j=1}^m p(y_j) \sum_{i=1}^n \sigma(i) p(x_i|y_j) \\
								& = \sum_{i=1}^n \sigma(i) \sum_{j=1}^m p(x_i,y_j) \\
								& = \sum_{i=1}^n \sigma(i) p(x_i) = \mathrm{G}(X),
			\end{align*}
		where the inequality is due to the fact that substituting $\sigma$ for $\sigma_j$ can only increase (at best preserve) the posterior guesswork $\mathrm{G}(X|Y=y_j)$. Moreover, the equality is achieved if and only if each prior optimal permutation $\sigma$ is optimal with respect to any posterior distribution $\{p(x_i|y_j): 1\leq i\leq n\}$. (Since it may be the case that $p(x_i) = p(x_j)$ for some $1\leq i\not= j\leq n$, there may exist more than one optimal permutations achieving the prior minimum guesswork $\mathrm{G}(X)$.) Now, suppose that $\mathcal{E}=\{(p(x_i),\rho_{x_i})\}$ is a quantum encoding of $X$ such that $\mathrm{G}^{opt}(\mathcal{E}) = \mathrm{G}(X)$. As we need to range over all POVMs in obtaining $\mathrm{G}^{opt}(\mathcal{E})$, the aforementioned equality condition can be restated as
			\begin{align*}
				\forall & 1\leq i,j\leq n, \pi\in\mathcal{P}. \\
				& p(x_i) \geq p(x_j) \Rightarrow p(x_i)\mathrm{Tr}(\rho_{x_i}\pi) \geq p(x_j)\mathrm{Tr}(\rho_{x_j}\pi),
			\end{align*}
		which is equivalent to
			\begin{align*}
				\forall 1\leq i,j\leq n. \quad p(x_i) \geq p(x_j) \Rightarrow p(x_i)\rho_{x_i} \geq p(x_j)\rho_{x_j}.
			\end{align*}
		We formalize our result as the following theorem.
			\vspace*{12pt}
			\begin{mythm}
				\label{thm_n_m}
				Let $\mathcal{E} = \{(p(x_i),\rho_{x_i}): 1\leq i\leq n\}$ be a quantum encoding of a random variable $X$. Let $\mathrm{G}(X)$ and $\mathrm{G}^{opt}(\mathcal{E})$ be defined in Eq.(\ref{eq_def_guesswork}) and Eq.(\ref{eq_def_min_guesswork}), respectively. Then $\mathrm{G}^{opt}(\mathcal{E}) = \mathrm{G}(X)$ holds if and only if for any $1\leq i,j\leq n$ the following condition holds:
					\begin{align*}
						p(x_i) \geq p(x_j) \Rightarrow p(x_i)\rho_{x_i} \geq p(x_j)\rho_{x_j}.
					\end{align*}
				Consequently, it holds that 
						$\max_{\mathcal{E}} \mathrm{G}^{opt}(\mathcal{E}) = \mathrm{G}(X)$ 
				for any random variable $X$.
			\end{mythm}
			\vspace*{12pt}
		A direct corollary of the theorem is that, for a uniformly distributed variable $X$, a quantum encoding $\mathcal{E}$ achieving $\mathrm{G}^{opt}(\mathcal{E}) = \mathrm{G}(X)$ must be formed by identical states.

	\section{Discussion}
	\noindent
		\label{sec_discussion}
				By Proposition~\ref{prop_n!}, we have shown that the optimal POVM $\Pi$ achieving minimum guesswork can always be taken as an $n!$-component measurement when the ensemble $\mathcal{E}$ comprises $n$ states. Still, better upper bound on $|\Pi|$ may be possible, especially for cases where the quantum states have certain type of symmetry. Recall that in Example~\ref{exa_trine} we have found the optimal POVM $\Pi^G$, which consists of three elements, for the trine ensemble. 


				In Section~\ref{sec_bound}, we combined information-theoretic bounds on guesswork $\mathrm{G}(X|Y)$ in classic setting and bounds on accessible information $\mathrm{I}_{acc}(\mathcal{E})$ in the quantum setting to generate bounds on minimum guesswork $\mathrm{G}^{opt}(\mathcal{E})$. There are alternative approaches to deriving bounds on minimum guesswork. Recall that in Section~\ref{sec_g_e} we bounded $\mathrm{G}^{opt}(\mathcal{E})$ in terms of $\mathrm{P}^{opt}_{err}(\mathcal{E})$ in both directions (see Eq.(\ref{eq_thm_gp})). Thus, we can substitute existing bounds on minimum error probability~\cite{Barnum2002,Montanaro2007,Montanaro2008,Qiu2008,QiuL2010} for $\mathrm{P}^{opt}_{err}(\mathcal{E})$ in Eq.(\ref{eq_thm_gp}) to generate bounds on $\mathrm{G}^{opt}(\mathcal{E})$. Also, we can relate $\mathrm{G}^{opt}(\mathcal{E})$ to the inconclusive probability $p_{inc}$ of an unambiguous discrimination scheme (if applicable) in the following way:
					\begin{align*}
						\mathrm{G}^{opt}(\mathcal{E}) \leq (1-p_{inc}) + \frac{n+1}{2} p_{inc} = 1 + \frac{n-1}{2} p_{inc},
					\end{align*}
				where $n$ is the number of states in $\mathcal{E}$. This inequality follows from the fact that $\mathrm{G}(X)\leq (n+1)/2$~\cite{Santis2001}. Then, substituting the existing upper bound on $p_{inc}$~\cite{Duan1998,Sun2002} for $p_{inc}$ in the above inequality, we obtain an upper bound on $\mathrm{G}^{opt}(\mathcal{E})$. Comparing the performance of various bounds yielded by these approaches is a subject for future work. 
	
	\section{Conclusion}
	\noindent
		\label{sec_conclusion}
		In this work, with the concern of brute-force adversarial strategy, we reexamined the problem of quantum state discrimination by adopting guesswork~\cite{Massey1994,Arikan1996}, rather than the widely accepted and well-studied error probability~\cite{Helstrom1976}, as the optimization criterion. The new problem, named quantum guesswork discrimination, can thus be viewed as a sibling of minimum error discrimination. Following the approach in~\cite{EldarMV2003}, we reduced the new optimization problem to a semidefinite programming problem. Necessary and sufficient conditions which must be satisfied by the optimal POVM to achieve minimum guesswork were also presented. Then we investigated the relation between minimum guesswork and minimum error probability, showing that the former can be bounded in terms of the latter in both directions. Additionally, the general disagreement between the two criteria was illustrated by the trine ensemble. Combining Massey's and McEliece and Yu's classical bounds on guesswork~\cite{Massey1994,McelieceY1995} with two elegant bounds on accessible information~\cite{Holevo1973b,JozsaRW1994}, we obtained both upper and lower bounds on minimum guesswork. Other approaches to deriving bounds were also discussed in Section~\ref{sec_discussion}. For geometrically uniform quantum states~\cite{EldarF2001}, we gave sufficient conditions for a geometrically uniform POVM to achieve minimum guesswork. Using this result, we proved the optimality of the POVM $\Pi^G$ (see Example~\ref{exa_trine}) for the trine ensemble. Furthermore, inspired by a similar result in minimum error discrimination~\cite{Hunter2003}, we provided the necessary and sufficient condition under which making no measurement at all would be the optimal strategy.

	\section*{Acknowledgements}
		\noindent
		The work is supported by the Australian Research Council (Grant Nos. DP130102764 and FT100100218) and the National Natural Science Foundation of China (Grant Nos. 61170299 and 61370053). Y. Feng is also supported by the Overseas Team Program of the Academy of Mathematics and Systems Science, CAS and the CAS/SAFEA International Partnership Program for Creative Research Team.
	
		\noindent
		\setlength{\bibspacing}{\baselineskip}
		\bibliographystyle{unsrt}
		\bibliography{MGD}


	
\end{document}